\PassOptionsToPackage{dvipsnames}{xcolor}
\documentclass[10pt,twocolumn]{extarticle}
\usepackage[margin=0.75in,columnsep=2pc]{geometry}
\usepackage[numbers,square]{natbib}
\usepackage{amsmath,amsthm,booktabs,graphicx}
\usepackage[T1]{fontenc}
\usepackage[tt=false,type1=true]{libertine}
\usepackage[varqu]{zi4}
\usepackage[libertine]{newtxmath}
\usepackage{microtype}
\usepackage{hyperref}
\newif\ifarxiv
\usepackage{amsthm}
\usepackage{thmtools}
\usepackage{array}
\usepackage{float}
\usepackage{enumitem}
\usepackage[nameinlink,capitalise,noabbrev]{cleveref}
\makeatletter
\AddToHook{cmd/appendix/before}{\def\cref@section@alias{appendix}\def\cref@subsection@alias{appendix}}
\makeatother
\usepackage{mathtools}
\usepackage{tikz}
\usepackage{tcolorbox}
\usepackage{style}

\newtheorem{theorem}{Theorem}

\newtheorem{corollary}{Corollary}
\newtheorem{lemma}{Lemma}
\newcounter{definition}
\theoremstyle{definition}

\theoremstyle{remark}

\newcommand{\definitionboxtitle}[1]{Definition~\thedefinition\if\relax\detokenize{#1}\relax\else: #1\fi}
\newenvironment{definition}[1][]{\refstepcounter{definition}\begin{tcolorbox}[
    colback=MidnightBlue!4,
    colframe=MidnightBlue!70!black,
    colbacktitle=MidnightBlue!70!black,
    coltitle=white,
    fonttitle=\normalfont\small\bfseries,
    title={\definitionboxtitle{#1}},
    rounded corners,
    left=6pt,
    right=6pt,
    top=3pt,
    bottom=3pt,
    toptitle=3pt,
    bottomtitle=3pt,
    before skip=8pt,
    after skip=8pt,
    before upper={
  \setlength{\abovedisplayskip}{4pt}\setlength{\belowdisplayskip}{4pt}\setlength{\parskip}{4pt}}
  ]}{\end{tcolorbox}}
\crefname{definition}{Definition}{Definitions}
\Crefname{definition}{Definition}{Definitions}

\renewcommand{\ge}{\geqslant}
\renewcommand{\geq}{\geqslant}
\renewcommand{\le}{\leqslant}
\renewcommand{\leq}{\leqslant}
\numberwithin{equation}{section}
\allowdisplaybreaks

\DeclarePairedDelimiter{\paren}{(}{)}

\newcommand{\Dist}[2][]{\mathscr{D}_{#1}\paren*{#2}}
\usetikzlibrary{arrows.meta}
\newcommand{\selfloop}[1]{%
  \tikz[baseline=(v.base)]{%
    \node[inner sep=0pt] (v) {$#1$};
    \draw[overlay, -{Stealth[length=2pt,width=2.2pt]}, line width=0.4pt]
      ([xshift=-0.1em,yshift=0.06em]v.north)
        .. controls +(115:0.6em) and +(65:0.6em) ..
      ([xshift=0.1em,yshift=0.06em]v.north);
    \path ([yshift=0.5em]v.north); 
  }}

\arxivtrue
\providecommand{\Description}[2][]{}
\hypersetup{linkcolor=black}

\usepackage{titling}
\pretitle{\begin{center}\huge\sffamily\bfseries}
\posttitle{\par\end{center}}
\usepackage{caption}
\usepackage{titlesec}
\titlespacing*{\section}{0pt}{.75\baselineskip plus 2pt minus .2pt}{.25\baselineskip}
\titlespacing*{\subsection}{0pt}{.75\baselineskip plus 2pt minus .2pt}{.25\baselineskip}
\title{Deterministic Metric Distortion of Top-2 Voting}
\author{Alec Sun\\University of Chicago \and Daniel Zhu\\The Harker School}
\date{}

\begin{document}
\maketitle

\begin{abstract}
We give tight bounds on the metric distortion of top-2 deterministic voting mechanisms. We first exploit the restrictive structure of top-2 voting to reduce the space of worst-case metrics from the continuum of biased metrics to $2^n$ subset metrics, where every candidate has distance $0$ or $1$ to the optimal candidate and $n$ is the number of candidates. We then reduce these further to just $2n$ polar metrics: $n$ of these have a single optimal candidate, and the other $n$ have a single sub-optimal candidate. This reduction yields an efficient algorithm that computes the exact distortion of every candidate on a given profile and runs in $\mathcal{O}(n^2)$ time. As a corollary, we can derive the exact worst-case distortion for top-2 voting for all values of $n$, for which previous work had non-matching lower and upper bounds.

\end{abstract}

\noindent\textbf{Keywords:} Computational social choice, Metric distortion, Voting, Top-$k$ ballots

\section{Introduction}

Social choice theory studies the aggregation of individual preferences into collective decisions. A voting mechanism is an algorithm that takes individual votes as input and tries to elect the optimal candidate, either deterministically or randomly. However, the definition of ``optimal'' here is unclear. Historically, researchers have tried to axiomatize desirable properties of an ideal voting system, such as Pareto efficiency, independence of irrelevant alternatives, non-dictatorship, and strategyproofness, but they have discovered that many of these axioms are mathematically incompatible~\citep{arrow1951,gibbard1973,satterthwaite1975}. 

Therefore, in order to better quantify and compare the effectiveness of different voting mechanisms, \citet{procaccia2006} introduced the \emph{distortion} framework, which reframes mechanism design as a limited-information optimization problem. In this framework, voters associate latent cardinal costs with each candidate but reveal only their ordinal rankings. The mechanism, which only sees these rankings, must then optimize its choice over all consistent cost configurations. In particular, it must minimize \emph{distortion}, defined as the cost ratio between the selected candidate and the optimal candidate. 

Additionally, sensible restrictions must be placed on these cost configurations in order to derive any meaningful results. For instance, consider a 2-voter, 2-candidate election in which voter $A$ prefers candidate $1$ to $2$, and voter $B$ prefers $2$ to $1$. By symmetry, our mechanism has no incentive to pick one candidate over the other, so suppose it selects candidate $1$. If the costs were configured as shown in \cref{tab:unrestricted}, we would have elected a candidate with cost $1$ while the optimal candidate has cost $\epsilon$, leading to a distortion of $1/\epsilon$. As $\epsilon \rightarrow 0$, this unrestricted cost assignment results in infinite distortion for every mechanism.

\begin{table}[htb]
\centering
\caption{Unrestricted costs consistent with the rankings. Electing candidate $1$ has distortion $1/\epsilon \to \infty$.}
\label{tab:unrestricted}
\begin{tabular}{@{}lccc@{}}
\toprule
 & Voter $A$ & Voter $B$ & Total \\
\midrule
Candidate $1$ & $0$ & $1$ & $1$ \\
Candidate $2$ & $\epsilon$ & $0$ & $\epsilon$ \\
\bottomrule
\end{tabular}

\end{table}

\emph{Metric distortion} explores one plausible restriction: the latent costs must form a metric. That is, the candidates and voters are embedded into a shared metric space, and the cost voter $v$ associates with candidate $i$ corresponds to their distance in this metric $d\paren*{v, i}$. Under this constraint, the distortion of our two-voter, two-candidate example collapses from $\infty$ to $3$, witnessed by the metric shown in \cref{fig:metric-distortion}. 

\begin{figure}[htb]
\centering
\begin{tikzpicture}[>=stealth, scale=1.05,
    cand/.style={circle, fill=black, inner sep=2.3pt},
    vot/.style={rectangle, fill=black, inner sep=2.3pt},
    elab/.style={font=\small, fill=white, inner sep=1pt}]
  \coordinate (c1) at (0, 0);
  \coordinate (vA) at (2, 1.1);
  \coordinate (c2) at (4, 0);
  \coordinate (vB) at (4, -1.5);
  \draw (c1) to[bend left=18]  node[elab] {$\frac{d}{2} - \epsilon$} (vA);
  \draw (vA) to[bend left=18]  node[elab] {$\frac{d}{2} + \epsilon$} (c2);
  \draw (c1) -- node[elab] {$d$} (c2);
  \draw (c1) to[bend right=30] node[elab] {$d$} (vB);
  \draw (c2) -- node[elab] {$0$} (vB);
  \node[cand] at (c1) {}; \node[cand] at (c2) {};
  \node[left=3pt,  red!80!black, font=\bfseries] at (c1) {$1$};
  \node[right=3pt, red!80!black, font=\bfseries] at (c2) {$2$};
  \node[vot] at (vA) {}; \node[vot] at (vB) {};
  \node[above=2pt, font=\small] at (vA) {$A$};
  \node[right=3pt, font=\small] at (vB) {$B$};
\end{tikzpicture}
\caption{A metric consistent with the rankings of voters $A$ and $B$, drawn as a graph whose shortest-path lengths are the distances. Electing candidate $1$ has distortion $\frac{3d/2 - \epsilon}{d/2 + \epsilon} \to 3$.}
\label{fig:metric-distortion}
\Description{A graph with candidates 1 and 2 and voters A and B. Candidate 1 is at distance d/2 minus epsilon from voter A, distance d from candidate 2, and distance d from voter B. Voter A is at distance d/2 plus epsilon from candidate 2, and voter B is at distance 0 from candidate 2.}
\end{figure}

Somewhat surprisingly, this simple lower bound of $3$ is also tight: \citet{gkatzelis2020} proved that with full rankings, the optimal deterministic mechanism's distortion does not exceed $3$, and \citet{kizilkaya2022} later attained the same bound with the simpler Plurality Veto rule.

However, although mechanisms could in principle solicit full ordinal rankings, asking every voter for a complete permutation of all $n$ candidates imposes a large cognitive and communication burden. 
Therefore, practical systems often implement reduced top-$k$ ballots that only require voters to rank their top $k$ choices. Top-$1$, or plurality, is the most prevalent in major elections due to its simplicity; but many elections employ richer rankings, such as New York City which uses top-$5$, or San Francisco which uses top-$10$. An interesting question to ask, then, is how much benefit fuller rankings provide. Concretely: how exactly does distortion decrease as we increase $k$?

\subsection{Our results}

We compute exact values for distortion when $k \leq 2$ and use them to close the gap between the existing lower and upper bounds of \citet{kempe2019} and \citet{anshelevich2018}.

\begin{restatable}{theorem}{kTwoExact} \label{thm:k2-exact}
For a $k \leq 2$ profile, let $\nu_c$ be the fraction of voters ranking candidate $c$ first, $t_c$ the fraction ranking $c$ anywhere on their ballot, and $\Pr_v\bigl[w \succ_v c\bigr]$ the fraction who prefer $w$ to $c$. Then the worst-case distortion of electing candidate $w$ is exactly
\[
\Dist{w} = 1 + 2 \max\paren[\Bigg]{\underbrace{\frac{1 - t_w}{t_w}}_{\text{approval}},\ \underbrace{\max_{c \neq w} \frac{1 - \Pr_v\bigl[w \succ_v c\bigr]}{1 - \nu_c}}_{\text{domination}}},
\]
so the best possible distortion of any $k \leq 2$ mechanism is $\min_w \Dist{w}$, attained by an $\mathcal{O}(n^2)$ sweep.
\end{restatable}

Zooming out from specific instances and considering worst case distortion as a function of $n$, we use \cref{thm:k2-exact} to derive exact per-$n$ bounds for $k = 2$ and all values of $n$. 
\begin{theorem} \label{thm:k2-worst}
For $k = 2$, the best possible worst-case distortion $W_2\paren*{n}$ of a deterministic top-$2$ mechanism is given by the last row of the following table.
\begin{center}
\setlength{\tabcolsep}{4pt}
\begin{tabular}{@{}lcccc@{}}
\toprule
$n$ & $3$ & $4$ & $5$ & $\geq 6$ \\
\midrule
Previous lower bound & $3$ & $3$ & $4$ & $n - 1$ \\
Previous upper bound & $3$ & $7$ & $9$ & $2n - 1$ \\
\addlinespace
Exact bound, $W_2\paren*{n}$ & $3$ & $1 + 2\sqrt{2}$ & $\sqrt{17}$ & $n - 1$ \\
\bottomrule
\end{tabular}
\end{center}
\end{theorem}

The previous lower bounds are $\max\paren*{3, n - 1}$, combining the full-ranking bound of \citet{anshelevich2018} with \citet{kempe2019}. For $n = 3$, top-$2$ ballots are full rankings, so $W_2\paren*{3} = 3$ already follows from \citet{gkatzelis2020}; for $n \geq 4$ the previous upper bound is plurality's $2n - 1$. We also note that the general top-$k$ upper bound of $\frac{6n}{k} + 1$ due to \citet{anagnostides2022} gives only $3n + 1$ at $k = 2$, which plurality's $2n - 1$ strictly dominates.

To derive our bounds, we start from the biased metrics of \citet{charikar2022} and exploit the structure of $k \leq 2$ to reduce them to a discrete set of $2^n$ subset metrics (the generalized \mbox{$\paren*{0, 1, 2, 3}$-metrics} of \citep{charikar2022}), then reduce those even further to just $2n$ polar metrics (\cref{lem:polar-metrics}). When $k > 2$ this chain of reduction no longer holds per-instance, but we conjecture that the $2n$ polar metrics may still suffice to compute the worst case distortion per-$n$.

\subsection{Related work}

\paragraph{Distortion.}
\citet{procaccia2006} introduced distortion to quantify the loss incurred when a mechanism sees only ordinal preferences rather than the underlying cardinal utilities, and \citet{boutilier2015} developed the view of voting rules as approximation algorithms for utilitarian social welfare.

\paragraph{Metric distortion with full rankings.}
\citet{anshelevich2018} introduced the metric model, showed that every deterministic rule has distortion at least $3$, that Copeland achieves $5$, and that plurality has distortion $2n - 1$. \citet{gkatzelis2020} closed the deterministic gap with Plurality Matching, which achieves distortion $3$, and \citet{kizilkaya2022} achieved the same bound with the simpler Plurality Veto rule.

\paragraph{Metric distortion with partial rankings.}
\citet{kempe2019} showed that any top-$k$ deterministic mechanism has distortion at least $\frac{2n - k}{k}$ and gave a generalization of Copeland to top-$k$ ballots with distortion $\mathcal{O}\paren*{n/k}$. \citet{anagnostides2022} sharpened the deterministic top-$k$ upper bound to $\frac{6n}{k} + 1$. Together with plurality's $2n - 1$~\citep{anshelevich2018}, these bounds leave a factor-$2$ gap at $k = 2$, which we close.

\paragraph{Randomized mechanisms.}
Randomization strictly helps in the metric setting. With full rankings, \citet{charikar2022} proved a lower bound of $2.1126$ and \citet{charikar2024} gave the first rule with distortion below $3$, namely $2.753$; the upper bound has since fallen quickly, with \citet{frank2026} and \citet{ye2026} independently reaching $\frac{5}{2}$ and \citet{shah2026} then reaching $2.1441$, nearly matching the lower bound. With partial ballots, \citet{gross2017} proved a $3 - \frac{2}{n}$ lower bound for randomized first-choice mechanisms and a $3 - \frac{2}{\lfloor n/k \rfloor}$ lower bound for randomized top-$k$ mechanisms with $k < \frac{n}{2}$, and \citet{kempe2019} gave a randomized first-choice mechanism attaining $3 - \frac{2}{n}$.

\section{Preliminaries} \label{sec:prelims}

Let $\mathcal{C}$ be a finite set of $n$ candidates and $V$ a finite set of voters, and fix $1 \leq k \leq n$. Each voter $v \in V$ reports a \emph{top-$k$ ballot}, an ordered list of $k$ distinct candidates that they claim to be their $k$ most preferred, from most to least preferred. The ballot induces a strict partial order $\succ_v$ on $\mathcal{C}$: $a \succ_v b$ if $v$ ranks both $a$ and $b$ with $a$ above $b$, or ranks $a$ but not $b$. Candidates that $v$ leaves unranked are incomparable under $\succ_v$. 

The \emph{profile} is the collection $\sigma$ of all ballots. For $k \leq 2$ we summarize $\sigma$ by $\nu_c$, the fraction of voters ranking $c$ first; $t_c$, the fraction ranking $c$ anywhere on their ballot; and $\Pr_v\bigl[w \succ_v c\bigr]$, the fraction of voters $v$ with $w \succ_v c$. 

A \emph{metric instance} places the voters and candidates in a common metric space with distance $d$.\footnote{As is standard, we allow distinct points at distance $0$, i.e.\ $d$ may be a pseudometric. Any pseudometric can be perturbed into a metric with an arbitrarily small change in every social cost, so this does not affect distortion.} It is \emph{consistent} with $\sigma$ if every voter's distances agree with their ballot: that is, $a \succ_v b$ implies $d\paren*{v, a} \leq d\paren*{v, b}$ for all $v \in V$ and $a, b \in \mathcal{C}$. The \emph{social cost} of a candidate $c$ is its expected distance to a random voter:
\[
C_d\paren*{c} = \bE_v\bigl[d\paren*{v, c}\bigr].
\]

\begin{definition}[Distortion] \label{def:distortion}
For a profile $\sigma$ and a candidate $w$, define
\[
\Dist[\sigma]{w} = \sup_{d \text{ consistent with } \sigma} \frac{C_d\paren*{w}}{\min_{c \in \mathcal{C}} C_d\paren*{c}},
\]
with $0/0 = 1$ and $z/0 = +\infty$ for $z > 0$. We drop the subscript when the profile is fixed.

A \emph{deterministic top-$k$ mechanism} $f$ maps each top-$k$ profile $\sigma$ to a single candidate $f\paren*{\sigma}$, and its \emph{distortion} is $\sup_\sigma \Dist[\sigma]{f\paren*{\sigma}}$. Note that the mechanism does not know the underlying metric.
\end{definition}

A mechanism may choose its winner separately for each profile, so the best possible distortion on a fixed profile $\sigma$ is $\min_w \Dist[\sigma]{w}$. Consequently, the best possible worst-case distortion of a deterministic top-$k$ mechanism with $n$ candidates is
\[
W_k\paren*{n} = \inf_f \sup_\sigma \Dist[\sigma]{f\paren*{\sigma}} = \sup_\sigma \min_{w \in \mathcal{C}} \Dist[\sigma]{w},
\]
where $\sigma$ ranges over top-$k$ profiles on $n$ candidates with any finite number of voters. It is known that $W_1\paren*{n} = 2n - 1$, with the upper bound attained by plurality~\citep{anshelevich2018} and the matching lower bound due to \citet{kempe2019}, and \cref{thm:k2-worst} determines $W_2\paren*{n}$ exactly.

\section{Reducing metrics}

The supremum in \cref{def:distortion} ranges over every consistent metric, which is an unwieldy set to optimize over. We cut it down using the \emph{biased metrics} of \citet{charikar2022}, a family parameterized only by the distances from each candidate to the optimal candidate $c^*$.

\begin{definition}[Biased metrics] \label{def:biased}
Let $x = \paren*{x_c}_{c \in \mathcal{C}}$ be nonnegative reals with $x_{c^*} = 0$ for at least one $c^*$. Write $c \succeq_v b$ to mean $b = c$ or $c \succ_v b$. The \emph{biased metric} for $x$ places each voter $v$ at distance
\begin{gather*}
\hat{d}\paren*{v, c} = y_v + z_{v, c}, \qquad \text{where} \\
y_v = \frac{1}{2} \max\paren*{0,\ \max_{a \succ_v b} \paren*{x_a - x_b}} \quad \text{and} \quad z_{v, c} = \min_{b \,:\, c \succeq_v b} x_b,
\end{gather*}
from each candidate $c$, and sets $\hat{d}\paren*{c, c'} = \max\paren*{x_c, x_{c'}}$ for candidates $c \neq c'$.
\end{definition}

In other words, $y_v = \hat{d}\paren*{v, c^*}$ is voter $v$'s distance to the optimal candidate $c^*$, and $z_{v, w} = \hat{d}\paren*{v, w} - \hat{d}\paren*{v, c^*}$ is how much farther $v$ is from $w$ than from $c^*$. Once we fix all $x_c = d(c, c^*)$, i.e. the radial distances from $c^*$, biased metrics seek to greedily maximize $z_{v, w}$ while minimizing $y_v$. 

\begin{lemma} \label{lem:biased-valid}
For every $x$ and every top-$k$ profile $\sigma$, the biased metric for $x$ is a metric on $V \cup \mathcal{C}$ that is consistent with $\sigma$.
\end{lemma}

The point of the family is that it already contains the worst case.

\begin{lemma}[after Proposition 4 of \citet{charikar2022}] \label{lem:biased}
For any metric $d$ consistent with $\sigma$, there exists a biased metric $\hat{d}$ with equal or greater distortion.
\end{lemma}

Consequently the supremum defining $\Dist[\sigma]{w}$ is attained on biased metrics alone, and the only remaining freedom is the vector $x$.

\begin{corollary} \label{cor:biased}
For every top-$k$ profile $\sigma$ and every candidate $w$,
\[
\Dist[\sigma]{w} = \max_{c \in \mathcal{C}} \sup_{x \in X_c} \frac{C_{\hat{d}_x}\paren*{w}}{C_{\hat{d}_x}\paren*{c}},
\]
where $X_c = \bigl\{x \in \mathbb{R}_{\geq 0}^{\mathcal{C}} : x_c = 0\bigr\}$ and $\hat{d}_x$ is the biased metric for $x$.
\end{corollary}

\begin{corollary} \label{cor:better-dist}
For every top-$k$ profile $\sigma$ and every candidate $w$,
\[
\Dist[\sigma]{w} = 1 + \max_{c \in \mathcal{C}} \sup_{x \in X_c} \frac{\bE_v\bigl[z_{v, w}\bigr]}{\bE_v\bigl[y_v\bigr]},
\]
where here $0/0 = 0$ and $z/0 = +\infty$ for $z > 0$.
\end{corollary}

\begin{proof}
Since $x \geq 0$ and $x_c = 0$, we have $z_{v, c} = 0$ for every $v$, so $C_{\hat{d}_x}\paren*{c} = \bE_v\bigl[y_v\bigr]$ and $C_{\hat{d}_x}\paren*{w} = \bE_v\bigl[y_v + z_{v, w}\bigr]$. Substituting into \cref{cor:biased} gives the result.
\end{proof}

The proofs of these results are essentially the same as \citep{charikar2022}. The only difference is their consideration of full rankings versus our top-$k$ rankings, but this gap can be remedied with minimal changes to the proof. For completeness, we reproduce top-$k$ proofs for \cref{lem:biased-valid}, \cref{lem:biased}, and \cref{cor:biased} in \cref{app:biased}.

We now prove that when $k \leq 2$, we can further reduce biased metrics to subset metrics.

\begin{definition}[Subset metrics] \label{def:subset}
For a set $S \subsetneq \mathcal{C}$, the \emph{subset metric} $\hat{d}_S$ is the biased metric for $x = \mathbf{1}_S$, that is, $x_a = 1$ for $a \in S$ and $x_a = 0$ for $a \notin S$. Up to a factor of $2$, which does not affect distortion, these are the generalized $\paren*{0, 1, 2, 3}$-metrics of \citet{charikar2022}.
\end{definition}

\begin{lemma} \label{lem:subsets}
When $k \leq 2$, the supremum in \cref{cor:biased} is attained on subset metrics alone, that is, for every top-$k$ profile $\sigma$ and every candidate $w$,
\[
\Dist[\sigma]{w} = \max_{c \in \mathcal{C}} \max_{S \subseteq \mathcal{C} \setminus \{c\}} \frac{C_{\hat{d}_S}\paren*{w}}{C_{\hat{d}_S}\paren*{c}}.
\]
\end{lemma}

\begin{proof}
We show that when $k \leq 2$, any biased metric can be massaged into a subset metric with equal or greater distortion. This subset metric will have the same (weak) relative order among $x_i$ as the biased metric, but all $x_i$ will now be pushed to $1$ or $0$, and we will maintain $x_c = 0$. Since every subset metric $\hat{d}_S$ with $c \notin S$ is itself a biased metric in $X_c$, this suffices. So fix $c$ and $x \in X_c$.

Fix an order $\pi$ of $\mathcal{C}$ with $c$ first, restrict to $x$ that are nondecreasing along $\pi$, and break ties in every $\max$ and $\min$ according to $\pi$. Then each $z_{v, w} = \min_{b \,:\, w \succeq_v b} x_b$ collapses into a single term $x_{b^*}$, where $b^*$ is the $\pi$-earliest $b$ with $w \succeq_v b$; so $z_{v, w}$ is linear in $x$.

Furthermore, when $k \leq 2$, all $y_v$ also become linear in $x$. Consider the maximizing $(a, b)$ pair in $y_v = \frac{1}{2} \max\paren*{0,\ \max_{a \succ_v b} \paren*{x_a - x_b}}$, and let $U_v$ be the candidates $v$ leaves unranked. If $v$'s ballot is $(d, c)$ (recall $x_c = 0$), the maximizing pair will always be $(d, c)$, since the other pairs $(d, u)$ and $(c, u)$ with $u \in U_v$ have $x_u \geq 0 = x_c$. 
Otherwise, if $v$'s ballot is $(d, e)$, the maximizing pair will be $(\arg \max_{i \in \{d, e\}} x_i, \arg \min_{i \in U_v} x_i)$: the pairs in $\{d, e\} \times U_v$ are maximized by this choice, and the remaining pair $(d, e)$ is dominated by $(d, c)$ if $c \in U_v$, or by the clamp at $0$ if $d = c$. The same argument gives the pair $(d, \arg \min_{i \neq d} x_i)$ for a $k = 1$ ballot $(d)$. In each case the maximizing pair, and whether the clamp at $0$ is active, depends only on $v$ and $\pi$, so $2 y_v = x_p - x_q$ for all $x$ consistent with $\pi$, for fixed candidates $p, q$ (with $p = q$ allowed).

Thus, for $x$ consistent with $\pi$, both $z_{v, w}$ and $y_v$ are linear in $x$, which means the expression from \cref{cor:better-dist},
\[
\sup_{x \in X_c} \frac{\bE_v\bigl[z_{v, w}\bigr]}{\bE_v\bigl[y_v\bigr]},
\]
reduces to a supremum over ratios of nonnegative linear forms in $x$, namely $N\paren*{x} / D\paren*{x}$ with $N\paren*{x} = \bE_v\bigl[z_{v, w}\bigr]$ and $D\paren*{x} = \bE_v\bigl[y_v\bigr]$. Under the ordering constraint, this ratio is maximized when all $x_i$ equal either $0$ or $1$. Indeed, write $\pi = \paren*{\pi_0, \dots, \pi_{n-1}}$ with $\pi_0 = c$ and $S_j = \{\pi_j, \dots, \pi_{n-1}\}$. Every $x$ consistent with $\pi$ then decomposes as
\[
x = \sum_{j=1}^{n-1} \lambda_j \mathbf{1}_{S_j}, \qquad \lambda_j = x_{\pi_j} - x_{\pi_{j-1}} \geq 0,
\]
so by linearity of $N$ and $D$, and applying the mediant inequality,
\[
\frac{N\paren*{x}}{D\paren*{x}} = \frac{N\paren*{\sum_j \lambda_j \mathbf{1}_{S_j}}}{D\paren*{\sum_j \lambda_j \mathbf{1}_{S_j}}} = \frac{\sum_j \lambda_j N\paren*{\mathbf{1}_{S_j}}}{\sum_j \lambda_j D\paren*{\mathbf{1}_{S_j}}} \leq \max_{j \,:\, \lambda_j > 0} \frac{N\paren*{\mathbf{1}_{S_j}}}{D\paren*{\mathbf{1}_{S_j}}}.
\]
This also holds when $D\paren*{x} = 0$ under the conventions of \cref{cor:better-dist}: if $N\paren*{x} > 0$, some $j$ with $\lambda_j > 0$ has $D\paren*{\mathbf{1}_{S_j}} = 0 < N\paren*{\mathbf{1}_{S_j}}$, so the right-hand side is also $+\infty$, and if $N\paren*{x} = 0$ the left-hand side is $0$. Each $\mathbf{1}_{S_j}$ gives the subset metric $\hat{d}_{S_j}$ with $S_j \subseteq \mathcal{C} \setminus \{c\}$, and maximizing over $c$ completes the proof.
\end{proof}

We can continue exploiting the structure of $k \leq 2$ ballots to further reduce this set of $2^n$ subset metrics to just $2n$ polar metrics.

\begin{definition}[Cores] \label{def:cores}
The \emph{cores} of a subset metric $\hat{d}_S$ are the candidates $i$ with $x_i = 0$, that is, the candidates in $\mathcal{C} \setminus S$. An \emph{$m$-core metric} is a subset metric with exactly $m$ cores.
\end{definition}

\begin{lemma} \label{lem:core-critical}
Consider any top-$k$ profile $\sigma$. Among subset metrics with at least $k$ cores, the distortion of a candidate $w$ is maximized by the $\paren*{n - 1}$-core metric $\hat{d}_{\{w\}}$, where every candidate but $w$ is a core.
\end{lemma}

\begin{proof}
Fix a subset metric with at least $k$ cores. As in \cref{cor:better-dist}, the distortion of $w$ is
\[
\Dist{w} = 1 + \frac{\bE_v\bigl[z_{v, w}\bigr]}{\bE_v\bigl[y_v\bigr]}.
\]
If $w$ is a core, then $z_{v, w} = 0$ for all $v$ and the distortion is $1$, so we may assume $w$ is not a core, i.e.\ $x_w = 1$.

We first show that any metric with at least $k$ cores shares the same $\bE_v\bigl[z_{v, w}\bigr]$. In particular, we show that regardless of what the $k$ (or greater) cores are, $z_{v, w} = 0$ whenever $v$ ranks $w$ in \textit{any position}, and $1$ otherwise. 

Recall that $z_{v, w} = \min_{b \,:\, w \succeq_v b} x_b$. If voter $v$ ranks $w$ anywhere, they can rank only $k - 1$ other candidates. Since $w$ is not a core and there are $\geq k$ core candidates, at least one core must be unranked, and therefore provides a valid value of $b$ with $w \succeq_v b$; so $z_{v, w} = x_b = 0$. On the other hand, if $w$ is unranked, the $b$ that satisfies $w \succeq_v b$ is $b = w$. By assumption $x_w = 1$, so $z_{v, w} = x_b = x_w = 1$. Therefore, it holds for all $\geq k$-core metrics that $\bE_v\bigl[z_{v, w}\bigr] = 1 - t_w$, where $t_w$ equals the fraction of voters with $w$ anywhere on their ballot. 

Since the numerator is fixed, we just need to minimize the denominator. Note that $y_v = \frac{1}{2}$ whenever $a \succ_v b$ for some non-core $a$ and core $b$, and $0$ otherwise. Therefore, $\bE_v\bigl[y_v\bigr] \geq \frac{1}{2} t_w$, since every voter who ranks $w$ ranks it above an unranked core (as shown above) and $w$ is not a core. This holds with equality for $\hat{d}_{\{w\}}$, and therefore we have shown that $\hat{d}_{\{w\}}$ maximizes $\Dist{w}$, as desired. 
\end{proof}

\begin{definition}[Polar metrics] \label{def:polar}
A \emph{polar metric} is a subset metric with either $1$ or $n - 1$ cores.
\end{definition}

\begin{lemma} \label{lem:polar-metrics}
When $k \leq 2$, the supremum in \cref{cor:biased} is attained on polar metrics alone, that is, for every top-$k$ profile $\sigma$ and every candidate $w$,
\[
\Dist[\sigma]{w} = \max_{c \in \mathcal{C}} \max_{\substack{S \subseteq \mathcal{C} \setminus \{c\} \\ \left|S\right| \in \{1, n - 1\}}} \frac{C_{\hat{d}_S}\paren*{w}}{C_{\hat{d}_S}\paren*{c}}.
\]
\end{lemma}

\begin{proof}
By \cref{lem:subsets}, it suffices to consider subset metrics $\hat{d}_S$. Every subset metric has either exactly $1$ core or at least $2 \geq k$ cores; and by \cref{lem:core-critical}, among metrics with at least $k$ cores, the maximum distortion is attained at an $(n - 1)$-core metric $\hat{d}_{\{w\}}$. Therefore, the only metrics we need to consider are these $n$ $(n - 1)$-core metrics and the $n$ 1-core metrics, i.e. the polar metrics.  
\end{proof}

We can now prove \cref{thm:k2-exact}.

\kTwoExact*

\begin{proof}
By \cref{lem:polar-metrics} and \cref{cor:better-dist}, $\Dist{w}$ is $1$ plus the maximum of $\bE_v\bigl[z_{v, w}\bigr] / \bE_v\bigl[y_v\bigr]$ over polar metrics. We show that the approval term comes from the $\paren*{n - 1}$-core metrics, while the domination term comes from the $1$-core metrics.

\emph{$\paren*{n - 1}$-core metrics.} Suppose $n \geq 3$, so these metrics have $n - 1 \geq k$ cores. For $\hat{d}_{\{a\}}$ with $a \neq w$, $w$ is a core, so $z_{v, w} = 0$ for all $v$ and the ratio is $0$. For $\hat{d}_{\{w\}}$, the proof of \cref{lem:core-critical} gives $\bE_v\bigl[z_{v, w}\bigr] = 1 - t_w$ and $\bE_v\bigl[y_v\bigr] = \frac{1}{2} t_w$, so the ratio is $2 \cdot \frac{1 - t_w}{t_w}$, the approval term. (When $n = 2$, every polar metric is also a $1$-core metric and is handled below.)

\emph{$1$-core metrics.} Let $c$ be the only core, so $x_c = 0$ and $x_a = 1$ for all $a \neq c$. If $c = w$, the ratio is $0$, so assume $c \neq w$. Then $z_{v, w} = 0$ exactly when $w \succ_v c$, that is, when $v$ ranks $w$ and does not rank $c$ above it, so $\bE_v\bigl[z_{v, w}\bigr] = 1 - \Pr_v\bigl[w \succ_v c\bigr]$. Similarly, $y_v = \frac{1}{2}$ exactly when $a \succ_v c$ for some $a \neq c$, that is, when $v$ does not rank $c$ first, so $\bE_v\bigl[y_v\bigr] = \frac{1}{2} \paren*{1 - \nu_c}$. The ratio is therefore $2 \cdot \frac{1 - \Pr_v\bigl[w \succ_v c\bigr]}{1 - \nu_c}$, and maximizing over $c \neq w$ gives the domination term.

Taking the larger of the two terms gives the formula for $\Dist{w}$.
\end{proof}

When $k = 1$, we have $t_w = \nu_w$ and $\Pr_v\bigl[w \succ_v c\bigr] = \nu_w$, and since $1 - \nu_c \geq \nu_w$, the domination term never exceeds the approval term. So \cref{thm:k2-exact} reduces to $\Dist{w} = \frac{2}{\nu_w} - 1$, whose upper bound already follows from \citet[Lemma~3]{gross2017} applied to the mechanism that always elects $w$.

If we remove the domination term, we get a simpler lower bound on distortion:
\begin{corollary} \label{cor:approval-lower}
For all values of $k$, let $t_c$ denote the fraction of voters who rank $c$ anywhere on their ballot. Then, for every top-$k$ profile, the best possible distortion satisfies
\[
\min_{w \in \mathcal{C}} \Dist{w} \geq 1 + 2 \min_{w \in \mathcal{C}} \frac{1 - t_w}{t_w} = \frac{2}{\max_{w} t_w} - 1,
\]
where $\frac{1 - t_w}{t_w} = +\infty$ if $t_w = 0$.
\end{corollary}

\begin{proof}
This lower bound holds for all profiles regardless of $k$, simply by taking the $\paren*{n - 1}$-core metric $\hat{d}_{\{w\}}$ in \cref{cor:better-dist}, with any $c \neq w$ as the optimal candidate.
\end{proof}

For a fixed $n$, this lower bound is maximized at a uniform profile:
\begin{corollary} \label{cor:nk-lower}
Any top-$k$ deterministic voting mechanism on $n$ candidates has distortion at least $\frac{2n}{k} - 1$. 
\end{corollary}

\begin{proof}
  Label the candidates $0, \dots, n - 1$ and take $n$ voters, where voter $i$ reports the ballot $(i, i + 1, \dots, i + k - 1)$ with indices mod $n$. Each candidate appears on exactly $k$ ballots, so $t_w = \frac{k}{n}$ for every $w$, and \cref{cor:approval-lower} gives $\Dist{f\paren*{\sigma}} \geq \min_w \Dist{w} \geq \frac{2n}{k} - 1$ for any mechanism $f$.
\end{proof}

This recovers the lower bound of $\frac{2n - k}{k}$ due to \citet{kempe2019}.

\section{A simple tight bound for \texorpdfstring{$n > 5$}{n > 5}}

We construct a simple mechanism that always achieves distortion at most $\max(n - 1, 5)$, thereby matching the lower bound of \cref{cor:nk-lower} for $k = 2$ when $n > 5$. 

\begin{definition}[Majority-or-approval] \label{def:majority-or-approval}
If some candidate $w$ holds a majority of first-place votes, i.e.\ $\nu_w \geq \frac{1}{2}$, the \emph{majority-or-approval} mechanism elects $w$; otherwise, it elects a candidate with the highest approval $t_w$. Ties are broken arbitrarily.
\end{definition}

\begin{lemma} \label{lem:maj-or-approval-dist}
For $k = 2$, majority-or-approval achieves distortion at most $\max(n - 1, 5)$.  
\end{lemma}

\begin{proof}
By \cref{thm:k2-exact}, it suffices to bound the approval and domination terms of the elected candidate $w$ for every $c \neq w$. If there exists any candidate with $\nu_w \geq \frac{1}{2}$, we elect them. Since $t_w \geq \nu_w$, we have that the approval term $\frac{1 - t_w}{t_w} \leq \frac{1 - \nu_w}{\nu_w} \leq 1$. For the domination term, since every voter ranking $w$ first has $w \succ_v c$, $\Pr_v\bigl[w \succ_v c\bigr] \geq \nu_w$; together with $\nu_c \leq 1 - \nu_w$, we have $\frac{1 - \Pr_v\bigl[w \succ_v c\bigr]}{1 - \nu_c} \leq \frac{1 - \nu_w}{1 - \nu_c} \leq \frac{1 - \nu_w}{\nu_w} \leq 1$. Therefore, the distortion in this case is bounded by $1 + 2 \max(1, 1) = 3$.

Otherwise, all candidates satisfy $\nu_w < \frac{1}{2}$, and we elect the candidate with maximum $t_w$. Since every ballot names $2$ candidates, $\sum_c t_c = 2$, so $t_w \geq \frac{2}{n}$, and the approval term is bounded by $\frac{n}{2} - 1$. Since $\nu_c < \frac{1}{2}$, the domination term is bounded by $\frac{1 - \Pr_v\bigl[w \succ_v c\bigr]}{1 - \nu_c} \leq \frac{1}{1 - \nu_c} < 2$, so we achieve distortion at most
\[
1 + 2 \max\paren*{\frac{n}{2} - 1, 2} = \max\paren*{n - 1, 5}.
\]

The maximum distortion over these two cases is $\max \paren*{3, \max(n - 1, 5)} = \max(n - 1, 5)$, as desired. 
\end{proof}

Combining this with our simple approval lower bound:
\begin{corollary} \label{cor:n5-tight}
When $n > 5$, the worst-case distortion of any optimal top-$2$ deterministic mechanism is $W_2\paren*{n} = n - 1$, and majority-or-approval attains it.  
\end{corollary}

\begin{proof}
With $k = 2$, \cref{cor:nk-lower} gives a lower bound of $\frac{2n}{k} - 1 = n - 1$ for every mechanism. When $n > 5$, \cref{lem:maj-or-approval-dist} gives a matching upper bound of $\max(n - 1, 5) = n - 1$.  
\end{proof}

\section{Tight bounds for \texorpdfstring{$n \leq 5$}{n <= 5}}

Whereas the worst-case instances for $n > 5$ utilize only the approval term, the domination term remains relevant when $n \leq 5$. The exact distortions for these cases are shown in \cref{tab:small-n}.

\begin{table}[htb]
\centering
\caption{Exact worst-case distortion of deterministic top-$2$ mechanisms for $n \leq 5$.}
\label{tab:small-n}
\begin{tabular}{@{}lcccc@{}}
\toprule
$n$ & $2$ & $3$ & $4$ & $5$ \\
\midrule
$W_2\paren*{n}$ & $3$ & $3$ & $1 + 2\sqrt{2}$ & $\sqrt{17}$ \\
$\lambda = \frac{W_2\paren*{n} - 1}{2}$ & $1$ & $1$ & $\sqrt{2}$ & $\frac{\sqrt{17} - 1}{2}$ \\
\bottomrule
\end{tabular}
\end{table}

To verify these values, we first enumerate for each candidate, which polar metric among the $n$ possible ($1$ $(n - 1)$-core metric, and $(n - 1)$ $1$-core metrics) maximizes its distortion. Then, for each of these $n^n$ \textit{patterns}, we run an LP to verify that regardless of profile, at least one candidate will always have distortion at most $D^*$ under its chosen polar metric. This suffices, since a profile where every candidate has distortion above $D^*$ would violate this for the pattern of maximizing polar metrics.

In particular, for a fixed pattern, recall from the proof of \cref{thm:k2-exact} that the chosen polar metric gives each candidate the distortion

\[
1 + 2 \frac{N_w(\sigma)}{D_w(\sigma)}
\]

where $N_w$ is $1 - t_w$ or $1 - \Pr_v\bigl[w \succ_v c\bigr]$, and $D_w$ is $t_w$ or $1 - \nu_c$. Viewing $\sigma$ as its vector of ballot fractions in the simplex $\Delta$ over the $n\paren*{n - 1}$ possible ballots, $N_w$ and $D_w$ are nonnegative linear functions on $\Delta$. Then, we can rewrite the condition $\Dist{w} \leq D^* = 1 + 2 \lambda$ as
\[
\frac{N_w(\sigma)}{D_w(\sigma)} \leq \lambda \iff g_w(\sigma) = N_w(\sigma) - \lambda D_w(\sigma) \leq 0.
\]

where $g_w$ is also linear in $\Delta$. Then, we want to show that for all $\sigma \in \Delta$, $\min_w g_w(\sigma) \leq 0$. Since each $g_w$ is linear, by LP duality, this is true iff there exist some non-negative weights $y_w$ such that $\sum y_w = 1$ and $\sum_w y_w g_w(\sigma) \leq 0$ for all $\sigma$. If we further let $\sigma_{i j}$ denote the profile where all voters have the exact ballot $(i, j)$, then we have:
\begin{align*}
\max_\sigma \min_w g_w(\sigma) \leq 0 &\iff \exists\, y : \max_{\sigma \in \Delta} \sum_w y_w g_w(\sigma) \leq 0 \\
&\iff \exists\, y : \max_{i \neq j \in \mathcal{C}} \sum_w y_w g_w(\sigma_{i j}) \leq 0.
\end{align*}   

The second link holds because the maximum of a linear function on the simplex $\Delta$ must be attained at one of its vertices, $\sigma_{ij}$.

Thus, we can define the following LP over $y \geq 0$ with $\sum_w y_w = 1$ and $\delta$: minimize $\delta$ subject to $n(n - 1)$ constraints of the form $\sum_w y_w g_w(\sigma_{i j}) \leq \delta$. Written out in full, the LP for a fixed pattern is
\[
\begin{array}{ll}
\text{minimize} & \delta \\[2pt]
\text{subject to} & \displaystyle\sum_{w \in \mathcal{C}} y_w \, g_w\paren*{\sigma_{i j}} \leq \delta \quad \text{for all } i \neq j, \\[10pt]
& \displaystyle\sum_{w \in \mathcal{C}} y_w = 1, \qquad y \geq 0,
\end{array}
\]
over the variables $y \in \mathbb{R}^{\mathcal{C}}$ and $\delta \in \mathbb{R}$. The coefficients of $g_w$ are
\[
g_w\paren*{\sigma_{i j}} = 1 - \paren*{1 + \lambda} \mathbf{1}\bigl[w \in \{i, j\}\bigr]
\]
if the pattern chooses the approval term for $w$, and
\[
g_w\paren*{\sigma_{i j}} = 1 - \mathbf{1}\bigl[w \succ_{i j} c\bigr] - \lambda \paren*{1 - \mathbf{1}\bigl[c = i\bigr]}
\]
if it chooses the domination term for $w$ against $c$, where $\succ_{i j}$ is the order induced by the ballot $\paren*{i, j}$. Let $\delta^*$ be its optimum. If the maximum $\delta^*$ over all patterns is exactly $0$, then by duality $\max_\sigma \min_w g_w(\sigma) = 0$ for the maximizing pattern and $\leq 0$ for the rest, which means $\lambda$ is tight and $W_2\paren*{n} = 1 + 2 \lambda = D^*$ exactly.

Before we run the LP, we can cut down the search space by observing that we don't need to consider patterns that contain cycles. Below, we write a pattern as $\selfloop{w}$ (a self-loop) if it chooses the approval term for $w$, and $w \mapsto c$ if it chooses the domination term for $w$ against $c$. We also use $w_1 \mapsto w_2 \mapsto \dots \mapsto w_L$ as shorthand for the pattern $w_1 \mapsto w_2, \dots, w_{L - 1} \mapsto w_L$.
\begin{lemma} \label{lem:cycles}
If the pattern contains a cycle $w_1 \mapsto w_2 \mapsto \dots \mapsto w_L \mapsto w_1$ with $L \geq 2$, then for every profile $\sigma$, at least one candidate on the cycle has distortion at most 3 under its chosen polar metric, i.e. $\frac{N_{w_i}(\sigma)}{D_{w_i}(\sigma)} \leq 1$ for some $i$.
\end{lemma}

\begin{proof}
We show that $\frac{\sum_i N_{w_i}(\sigma)}{\sum_i D_{w_i}(\sigma)} \leq 1$. Substituting $N$ and $D$ and taking indices mod $L$, the numerator is $\sum_i 1 - \Pr_v\bigl[w_i \succ_v w_{i + 1}\bigr] = L - \sum_i \Pr_v\bigl[w_i \succ_v w_{i + 1}\bigr]$ while the denominator is $\sum_i 1 - \nu_{w_{i + 1}} = L - \sum_i \nu_{w_i}$. Because everyone who ranks $w_i$ first must rank $w_i$ above $w_{i + 1} \neq w_i$, $\Pr_v\bigl[w_i \succ_v w_{i + 1}\bigr] \geq \nu_{w_i}$ for all $i$, and therefore $\sum_i N_{w_i}(\sigma) \leq \sum_i D_{w_i}(\sigma) \implies \frac{\sum_i N_{w_i}(\sigma)}{\sum_i D_{w_i}(\sigma)} \leq 1$. 

Since all $N_{w_i}(\sigma)$ and $D_{w_i}(\sigma)$ are non-negative, it follows by mediant inequality that $\min_i \frac{N_{w_i}(\sigma)}{D_{w_i}(\sigma)} \leq 1$, as desired.
\end{proof}

We derive a similar bound on long paths.

\begin{lemma} \label{lem:pairs}
If the pattern contains $L \geq 2$ pairs $w_i \mapsto c_i$ with all $c_i$ distinct, then for every profile $\sigma$, at least one $w_i$ has distortion at most $3 + \frac{2}{L - 1}$ under its chosen polar metric, i.e. $\frac{N_{w_i}(\sigma)}{D_{w_i}(\sigma)} \leq \frac{L}{L - 1}$ for some $i$.
\end{lemma}

\begin{proof}
Again, we consider the mediant $\frac{\sum_i N_{w_i}(\sigma)}{\sum_i D_{w_i}(\sigma)}$ and prove it cannot exceed $\frac{L}{L - 1}$. For the numerator, 
$$
\sum_i N_{w_i}(\sigma) = \left(\sum_{i = 1}^{L} 1 - \Pr_v\bigl[w_i \succ_v c_i\bigr]\right) \leq L.
$$

For the denominator, since the $c_i$ are distinct, $\sum_i \nu_{c_i} \leq 1$, so
$$
\sum_i D_{w_i}(\sigma) = \left(\sum_{i = 1}^{L} 1 - \nu_{c_i} \right) = L - \sum_{i = 1}^{L} \nu_{c_i} \geq L - 1.
$$

Therefore $\frac{\sum_i N_{w_i}(\sigma)}{\sum_i D_{w_i}(\sigma)} \leq \frac{L}{L - 1}$, and applying the mediant inequality as in the proof of \cref{lem:cycles} (the denominator is at least $L - 1 \geq 1$) gives us the desired result.
\end{proof}

\begin{corollary} \label{cor:paths}
If the pattern contains a path $w_1 \mapsto w_2 \mapsto \dots \mapsto \selfloop{w_L}$ with $L \geq 3$, then for every profile $\sigma$, at least one $w_i$ with $i < L$ has distortion at most $3 + \frac{2}{L - 2}$ under its chosen polar metric, i.e. $\frac{N_{w_i}(\sigma)}{D_{w_i}(\sigma)} \leq \frac{L - 1}{L - 2}$ for some $i < L$.
\end{corollary}

\begin{proof}
Apply \cref{lem:pairs} to the $L - 1 \geq 2$ pairs $w_i \mapsto w_{i + 1}$.
\end{proof}

Finally, we can derive a bound for $L = 2$, where \cref{cor:paths} does not apply.

\begin{lemma} \label{lem:path-2}
If the pattern contains a path $w \mapsto \selfloop{c}$, then for every profile $\sigma$, at least one of $w, c$ has distortion at most $5$ under its chosen polar metric, i.e. $\frac{N_{w^*}(\sigma)}{D_{w^*}(\sigma)} \leq 2$ for some $w^* \in \{w, c\}$.
\end{lemma}

\begin{proof}
We again consider the mediant and bound the numerator and denominator:
\begin{align*}
N_w(\sigma) + N_c(\sigma) &= 1 - \Pr_v\bigl[w \succ_v c\bigr] + 1 - t_c \leq 2 \\
D_w(\sigma) + D_c(\sigma) &= 1 - \nu_c + t_c \geq 1
\end{align*}
where the second line uses $t_c \geq \nu_c$. So the mediant does not exceed $\frac{2}{1} = 2$, and applying the mediant inequality yields the desired result.
\end{proof}

Since even a single pair $w \mapsto c$ caps distortion at $5$, the simple approval bound of \cref{cor:nk-lower} dominates whenever $D^* > 5$, thereby giving another proof that when $n > 5, W_2(n) = n - 1$.

Since a lower bound of $3$ applies to all deterministic mechanisms regardless of $n$ and $k$~\citep{anshelevich2018}, any patterns containing cycles become redundant by \cref{lem:cycles}. In general, \cref{lem:pairs}, \cref{cor:paths}, and \cref{lem:path-2} show that all patterns containing domination terms have low, bounded distortion. Using these bounds to prune our pattern search space, we can now certify exact distortion values for $n \leq 5$:
\begin{lemma} \label{lem:small-n}
For $n \leq 5$, the worst-case distortions of deterministic top-2 mechanisms are exactly those shown in \cref{tab:small-n}. 
\end{lemma}

\begin{proof}
We simply enumerate all acyclic patterns (up to isomorphism suffices) and run the LP over all of them. For completeness, we provide exact values of $y$ for $n = 4, 5$ and all acyclic patterns up to isomorphism in \cref{app:certificates}.

$n = 2$ and $n = 3$ correspond to full rankings, so we already know these cases have distortion exactly $3$.

The tight pattern for $n = 4$ is $3 \mapsto \selfloop{1},\ 4 \mapsto \selfloop{2}$, maximized by the profile
\[
\begin{gathered}
\paren*{1, 3} : \sqrt{2} - 1, \qquad \paren*{2, 4} : \sqrt{2} - 1, \\
\paren*{3, 4} : \tfrac{3 - 2\sqrt{2}}{2}, \qquad \paren*{4, 3} : \tfrac{3 - 2\sqrt{2}}{2},
\end{gathered}
\]
which, like the pattern, is invariant under swapping $1 \leftrightarrow 2$ and $3 \leftrightarrow 4$.

The tight pattern for $n = 5$ is $5 \mapsto \selfloop{1},\ \selfloop{2},\ \selfloop{3},\ \selfloop{4}$. Writing $\tau = \frac{1}{1 + \lambda} = \frac{\sqrt{17} - 1}{8}$, it is maximized by the profile
\begin{gather*}
\paren*{1, 5} : \tau, \qquad \paren*{a, 5} : \tfrac{2 - 5\tau}{3}, \qquad \paren*{a, b} : \tfrac{4\tau - 1}{6} \\
\text{for all } a \neq b \in \{2, 3, 4\},
\end{gather*}
which, like the pattern, is invariant under every permutation of $\{2, 3, 4\}$.
\end{proof}

Combining these bounds with \cref{cor:n5-tight}, we recover \cref{thm:k2-worst}.

\section{Discussion}

While each link of the reduction from biased metrics to subset metrics to polar metrics holds \textit{per-instance} when $k \leq 2$, none of these links hold when $k > 2$.\footnote{For example, take $n = 4$ and $k = 3$. The profile with ballots $\paren*{1, 2, 4}$, $\paren*{1, 4, 2}$, $\paren*{2, 1, 3}$ in proportions $3 : 5 : 5$ has $\Dist{3} \geq \frac{47}{5}$, attained by the biased metric for $x = \paren*{0, \frac{1}{2}, 1, \frac{1}{2}}$, whereas every subset metric gives candidate $3$ distortion at most $\frac{31}{5}$. The profile with ballots $\paren*{1, 2, 3}$, $\paren*{2, 3, 1}$, $\paren*{3, 2, 1}$ in proportions $3 : 1 : 1$ gives candidate $1$ distortion $\frac{7}{3}$ under the subset metric $\hat{d}_{\{1, 4\}}$, whereas every polar metric gives it distortion at most $2$.} However, it remains open whether polar metrics suffice to compute the worst case \textit{per-$n$}: that is, if we compute $W_k\paren*{n} = \sup_\sigma \min_w \Dist[\sigma]{w}$ with each $\Dist[\sigma]{w}$ taken over only polar metrics, will this be the same value attained over all biased metrics? If true, we could generalize the tactics used to prove \cref{lem:small-n} to compute $W_k\paren*{n}$ for any fixed $n$ and $k$, since under each polar metric the numerator and denominator of \cref{cor:better-dist} remain linear in the ballot fractions.

When $k = 2$, we also showed in \cref{cor:n5-tight} that once $n > 5$, $W_2(n)$ exactly matches the approval lower bound $\frac{2n}{k} - 1$ of \cref{cor:nk-lower}. A natural next question is then whether such a threshold $n_k$ exists for all values of $k$, such that for all $n > n_k$, $W_k(n) = \frac{2n}{k} - 1$.  

\section*{AI Disclosure}

All proofs were derived and verified with the help of Claude Opus 5.5. The authors have reviewed all mathematical arguments, presented simplified versions, and retain full responsibility for any errors. The code provided in \ifarxiv the ancillary files\else the supplement\fi{} was also generated by Opus 5.5 and revised by the authors for concision and clarity.

\bibliography{references}

\appendix

\section{Proofs for biased metrics} \label{app:biased}

We follow \citet{charikar2022}, adapting each step to top-$k$ ballots. Fix a top-$k$ profile $\sigma$ and $x \geq 0$, and let $\hat{d}$ be the biased metric for $x$. We use two facts about every voter $v$ and candidate $a$. First, $a \succeq_v a$, so $z_{v, a} \leq x_a$. Second,
\begin{equation} \label{eq:x-bound}
x_a \leq 2 y_v + z_{v, a},
\end{equation}
since $z_{v, a} = x_b$ for some $b$ with $a \succeq_v b$, and either $b = a$ or $a \succ_v b$, in which case $x_a - x_b \leq 2 y_v$.

\begin{proof}[Proof of \cref{lem:biased-valid}]
\emph{Consistency.} If $a \succ_v b$, then by transitivity every $b'$ with $b \succeq_v b'$ also has $a \succ_v b'$, so $z_{v, a} \leq z_{v, b}$ and hence $\hat{d}\paren*{v, a} \leq \hat{d}\paren*{v, b}$.

\emph{Triangle inequalities with at most one voter.} For three candidates, $\max\paren*{x_a, x_b} \leq \max\paren*{x_a, x_e} + \max\paren*{x_e, x_b}$. For candidates $a \neq b$ and a voter $v$, \cref{eq:x-bound} and $y_v, z_{v, b} \geq 0$ give $x_a \leq \hat{d}\paren*{v, a} + \hat{d}\paren*{v, b}$, and likewise for $x_b$, so $\hat{d}\paren*{a, b} \leq \hat{d}\paren*{v, a} + \hat{d}\paren*{v, b}$. Finally, $z_{v, a} \leq x_a \leq \hat{d}\paren*{a, b}$ gives $\hat{d}\paren*{v, a} \leq \hat{d}\paren*{v, b} + \hat{d}\paren*{b, a}$.

\emph{Voters.} Distances between voters do not enter any social cost, so we set $\hat{d}\paren*{v, v'} = \min_{a} \bigl(\hat{d}\paren*{v, a} + \hat{d}\paren*{a, v'}\bigr)$ for $v \neq v'$. Every triangle inequality involving two or three voters follows by expanding each voter--voter distance through its minimizing candidate and applying the inequalities above.
\end{proof}

\begin{proof}[Proof of \cref{lem:biased}]
Let $c$ minimize $C_d$, and let $\hat{d}$ be the biased metric for $x_a = d\paren*{a, c}$, so $x \in X_c$. For every voter $v$:
\begin{itemize}
\item $y_v \leq d\paren*{v, c}$: if $a \succ_v b$, then $d\paren*{v, a} \leq d\paren*{v, b}$, so $x_a - x_b \leq \bigl(d\paren*{v, a} + d\paren*{v, c}\bigr) - \bigl(d\paren*{v, b} - d\paren*{v, c}\bigr) \leq 2 d\paren*{v, c}$.
\item $d\paren*{v, w} - d\paren*{v, c} \leq z_{v, w}$: if $w \succeq_v b$, then $d\paren*{v, w} \leq d\paren*{v, b} \leq d\paren*{v, c} + x_b$.
\end{itemize}
Averaging gives $C_{\hat{d}}\paren*{c} = \bE_v\bigl[y_v\bigr] \leq C_d\paren*{c}$ and $C_d\paren*{w} - C_d\paren*{c} \leq \bE_v\bigl[z_{v, w}\bigr] = C_{\hat{d}}\paren*{w} - C_{\hat{d}}\paren*{c}$. Since this last quantity is nonnegative,
\begin{equation} \label{eq:biased-ratio}
\frac{C_d\paren*{w}}{C_d\paren*{c}} \leq 1 + \frac{C_{\hat{d}}\paren*{w} - C_{\hat{d}}\paren*{c}}{C_d\paren*{c}} \leq \frac{C_{\hat{d}}\paren*{w}}{C_{\hat{d}}\paren*{c}},
\end{equation}
which also holds under the conventions of \cref{def:distortion} when $C_d\paren*{c} = 0$. The left-hand side is the distortion of $d$ for $w$, and the right-hand side is at most the distortion of $\hat{d}$ for $w$.
\end{proof}

\begin{proof}[Proof of \cref{cor:biased}]
For $x \in X_c$, $\hat{d}_x$ is consistent with $\sigma$ by \cref{lem:biased-valid}, so
\[
\Dist[\sigma]{w} \geq \frac{C_{\hat{d}_x}\paren*{w}}{\min_{c'} C_{\hat{d}_x}\paren*{c'}} \geq \frac{C_{\hat{d}_x}\paren*{w}}{C_{\hat{d}_x}\paren*{c}}.
\]
 Conversely, for every consistent $d$, the proof of \cref{lem:biased} gives some $c$ and $x \in X_c$ for which \cref{eq:biased-ratio} bounds the distortion of $d$ by $C_{\hat{d}_x}\paren*{w} / C_{\hat{d}_x}\paren*{c}$.
\end{proof}

\section{Certificates for \texorpdfstring{$n \leq 5$}{n <= 5}} \label{app:certificates}

For $n = 4$ and $n = 5$, we list a certificate $y$ for the LP of \cref{lem:small-n} for every acyclic pattern up to isomorphism, with $\lambda = \frac{W_2\paren*{n} - 1}{2}$ as in \cref{tab:small-n}. We write a pattern as the word $f\paren*{1} f\paren*{2} \cdots f\paren*{n}$, where $f\paren*{w} = w$ means the pattern chooses the approval term for $w$ and $f\paren*{w} = c \neq w$ means it chooses the domination term for $w$ against $c$. Two patterns are isomorphic if one is obtained from the other by relabeling the candidates, and we list the lexicographically smallest word in each orbit; a certificate for any other pattern in the orbit is obtained by relabeling $y$ accordingly. Each listed $y$ is invariant under every relabeling that fixes its pattern. Each $y = \paren*{y_1, \dots, y_n}$ satisfies $y \geq 0$, $y \neq 0$, and $\sum_w y_w \, g_w\paren*{\sigma_{i j}} \leq 0$ for every ballot $\paren*{i, j}$, which we verified in exact arithmetic over $\mathbb{Q}\paren*{\lambda}$. Daggers mark the tight orbits, where the LP optimum is $\delta^* = 0$; for $n = 4$ and $n = 5$ these are the orbits of the tight patterns in the proof of \cref{lem:small-n}.

In \cref{tab:cert-4,tab:cert-5}, each $y$ is scaled by an arbitrary positive constant and need not sum to $1$; since the constraints $y \geq 0$ and $\sum_w y_w \, g_w\paren*{\sigma_{i j}} \leq 0$ are invariant under scaling, the unnormalized $y$ is already a valid certificate, and dividing it by $\sum_w y_w$ gives a feasible solution to the LP with $\delta \leq 0$. Zero entries are shown as dots. Intuitively for each pattern, ``better'' candidates are assigned higher weight: for instance in the pattern $1133$, $y_1 = y_3 > y_2 = y_4$. 

\ifarxiv The ancillary files\footnote{Also available at \url{https://github.com/xyzqm/top2-voting}.} include\else The supplementary material includes\fi{} machine-readable certificates (\texttt{certificates.json}), an exact verifier (\texttt{verify\_certificate.py}), and a numerical solver for the LP (\texttt{solve\_lp.py}) that recomputes $\delta^*$, bisects to find the smallest $\lambda$ per-pattern, and also includes a function \texttt{W(n, k)} that maximizes this $\lambda$ over all patterns with $n$ candidates and top-$k$ ballots. When $k = 2$ this recomputes $W_2\paren*{n}$; for $k > 2$, it instead computes the lower bound on $W_k\paren*{n}$ given by polar metrics. However, we conjecture that this bound is tight.

\begin{table}[!t]
\centering
\caption{Certificates for $n = 4$, $\lambda = \sqrt{2}$.}
\label{tab:cert-4}
\begin{tabular}{@{}l w{c}{2.9em} w{c}{2.9em} w{c}{2.9em} w{c}{2.9em}@{}}
\toprule
$f$ & $y_{1}$ & $y_{2}$ & $y_{3}$ & $y_{4}$ \\
\midrule
$1\,1\,1\,1$ & $3$ & $1$ & $1$ & $1$ \\
$1\,1\,1\,2$ & $4$ & $5$ & $\cdot$ & $3$ \\
$1\,1\,1\,4$ & $4$ & $2$ & $2$ & $1$ \\
$1\,1\,2\,2$ & $8$ & $10$ & $3$ & $3$ \\
$1\,1\,2\,3$ & $\cdot$ & $3$ & $4$ & $5$ \\
\addlinespace[2pt]
$1\,1\,2\,4$ & $4$ & $5$ & $3$ & $\cdot$ \\
$1\,1\,3\,3^\dagger$ & $2 - \sqrt{2}$ & $\sqrt{2} - 1$ & $2 - \sqrt{2}$ & $\sqrt{2} - 1$ \\
$1\,1\,3\,4$ & $3$ & $2$ & $1$ & $1$ \\
$1\,2\,3\,4$ & $1$ & $1$ & $1$ & $1$ \\
\bottomrule
\end{tabular}

\medskip
\begingroup
\small
\setlength{\tabcolsep}{3pt}
\caption{Certificates for $n = 5$, $\lambda = \frac{\sqrt{17} - 1}{2}$.}
\label{tab:cert-5}
\begin{tabular}{@{}l w{c}{3.6em} w{c}{3.6em} w{c}{3.6em} w{c}{3.6em} w{c}{3.6em}@{}}
\toprule
$f$ & $y_{1}$ & $y_{2}$ & $y_{3}$ & $y_{4}$ & $y_{5}$ \\
\midrule
$1\,1\,1\,1\,1$ & $3$ & $1$ & $1$ & $1$ & $1$ \\
$1\,1\,1\,1\,2$ & $3$ & $3$ & $1$ & $1$ & $1$ \\
$1\,1\,1\,1\,5$ & $2$ & $1$ & $1$ & $1$ & $\cdot$ \\
$1\,1\,1\,2\,2$ & $8$ & $10$ & $\cdot$ & $3$ & $3$ \\
$1\,1\,1\,2\,3$ & $4$ & $3$ & $3$ & $1$ & $1$ \\
\addlinespace[2pt]
$1\,1\,1\,2\,4$ & $\cdot$ & $1$ & $\cdot$ & $1$ & $1$ \\
$1\,1\,1\,2\,5$ & $2$ & $2$ & $1$ & $1$ & $\cdot$ \\
$1\,1\,1\,4\,4$ & $9$ & $5$ & $5$ & $6$ & $5$ \\
$1\,1\,1\,4\,5$ & $4$ & $2$ & $2$ & $1$ & $1$ \\
$1\,1\,2\,2\,2$ & $11$ & $13$ & $3$ & $3$ & $3$ \\
\addlinespace[2pt]
$1\,1\,2\,2\,3$ & $\cdot$ & $1$ & $1$ & $\cdot$ & $1$ \\
$1\,1\,2\,2\,5$ & $8$ & $10$ & $3$ & $3$ & $\cdot$ \\
$1\,1\,2\,3\,3$ & $\cdot$ & $2$ & $2$ & $1$ & $1$ \\
$1\,1\,2\,3\,4$ & $\cdot$ & $\cdot$ & $1$ & $1$ & $1$ \\
$1\,1\,2\,3\,5$ & $\cdot$ & $1$ & $1$ & $1$ & $\cdot$ \\
\addlinespace[2pt]
$1\,1\,2\,4\,4$ & $1$ & $2$ & $1$ & $\cdot$ & $\cdot$ \\
$1\,1\,2\,4\,5$ & $1$ & $2$ & $1$ & $\cdot$ & $\cdot$ \\
$1\,1\,3\,3\,5$ & $1$ & $1$ & $1$ & $1$ & $\cdot$ \\
$1\,1\,3\,4\,5^\dagger$ & $\sqrt{17} - 1$ & $2\sqrt{17} - 6$ & $5 - \sqrt{17}$ & $5 - \sqrt{17}$ & $5 - \sqrt{17}$ \\
$1\,2\,3\,4\,5$ & $1$ & $1$ & $1$ & $1$ & $1$ \\
\bottomrule
\end{tabular}
\endgroup
\end{table}

\end{document}